\documentclass[a4paper,USenglish,cleveref,thm-restate]{lipics-v2021}

\pdfoutput=1 %
\hideLIPIcs  %
\usepackage[utf8]{inputenc}
\usepackage{todonotes}

\graphicspath{{img/}}

\usepackage{xfrac}

\title{Customizable Hub Labeling: Properties and Algorithms}

\author{Johannes Blum}{University of Konstanz, Germany \and \url{https://algo.uni-konstanz.de/team/blum}}{blum@inf.uni-konstanz.de}{https://orcid.org/0000-0003-1102-3649}{}

\author{Sabine Storandt}{University of Konstanz, Germany \and \url{https://algo.uni-konstanz.de/team/storandt}}{storandt@inf.uni-konstanz.de}{}{}

\authorrunning{J. Blum and S. Storandt} %

\Copyright{Johannes Blum and Sabine Storandt} %



\nolinenumbers %

\EventEditors{John Q. Open and Joan R. Access}
\EventNoEds{2}
\EventLongTitle{42nd Conference on Very Important Topics (CVIT 2016)}
\EventShortTitle{CVIT 2016}
\EventAcronym{CVIT}
\EventYear{2016}
\EventDate{December 24--27, 2016}
\EventLocation{Little Whinging, United Kingdom}
\EventLogo{}
\SeriesVolume{42}
\ArticleNo{23}
\DeclareMathOperator{\dist}{dist}

\newcommand{\bigO}{\mathcal{O}}
\newcommand{\card}[1]{\lvert #1 \rvert}

\newcommand{\searchspace}{\ensuremath{\mathit{SS}}}
\newcommand{\DSS}{\ensuremath{\mathit{DSS}}}

\newcommand{\Savg}{\ensuremath{S_{\mathrm{avg}}}}
\newcommand{\Lmax}{\ensuremath{L_{\max}}}
\newcommand{\Lavg}{\ensuremath{L_{\mathrm{avg}}}}

\newcommand{\Nup}{\ensuremath{N^{\uparrow}}}
\newcommand{\Ndown}{\ensuremath{N^{\downarrow}}}
\newcommand{\inv}{\ensuremath{_{\mathrm{inv}}}}

\newcommand{\nd}{\ensuremath{\mathrm{nd}}}

\newcommand{\inlineItem}[1]{\textcolor{lipicsGray}{\textbf{\textsf{#1}}}}
\begin{document}

\maketitle

\begin{abstract}
Hub Labeling (HL) is one of the state-of-the-art preprocessing-based techniques for route planning in road networks. It is a special incarnation of distance labeling, and it is well-studied in both theory and practice.  The core concept of HL is to associate a label with each vertex, which consists of a subset of all vertices and respective shortest path information, such that the shortest path distance between any two vertices can be derived from considering the intersection of their labels.  HL provides excellent query times but requires a time-consuming preprocessing phase. Therefore, in case of edge cost changes, rerunning the whole preprocessing is not viable.  Inspired by the concept of Customizable Route Planning, we hence propose in this paper a Customizable Hub Labeling variant for which the edge costs in the network do not need to be known at construction time. These labels can then be used with any edge costs after conducting a so called customization phase.  We study the theoretical properties of Customizable Hub Labelings, provide  an $\mathcal{O}(\log^2 n)$-approximation algorithm for the average label size, and propose efficient customization algorithms.
\end{abstract}

\section{Introduction}
Hub Labeling (HL) is one of the fastest algorithms to compute shortest paths in road networks. 
Given some weighted graph $G=(V,E)$, the idea is to assign a label $L(v) \subseteq V$ to every vertex $v$, such that for any two vertices $s$ and $t$ the \emph{cover property} is fulfilled, which states that the intersection $L(s) \cap L(t)$ of the respective labels contains a vertex $x$ on a shortest $s$-$t$-path.
Moreover for every \emph{hub} $x \in L(v)$ of $v$, the distance $\dist(v,x)$ from $v$ to $x$ is precomputed.
Exploiting the cover property, the distance between $s$ and $t$ can then simply be computed by identifying the vertex $x \in L(s) \cap L(t)$ that minimizes $\dist(s,x) + \dist(x,t)$.
If the labels are sorted, this can be done in time $\mathcal{O}(|L(s)|+|L(t)|)$. This approach allows for query times of less than a microsecond on country-sized road networks \cite{delling2013hub}; a speed-up of six orders of magnitude compared to a run of Dijkstra's algorithm.

However, one limitation of HL is that it is metric-dependent: If the weight of an edge changes, the shortest paths in $G$ may also change, and hence, the cover property might no longer be fulfilled. In this case, new labels have to be computed, which can be very time consuming (in the order of minutes or even hours \cite{Abraham2012}). This poses a serious problem for the applicability of HL; especially when considering travel times in road networks, which undergo frequent changes.
For other preprocessing-based route planning techniques, this problem has been combated by introducing the concept of customizability \cite{delling2011customizable}. Here, the preprocessing phase is divided into a metric-independent phase and the metric-dependent customization phase. The goal of the metric-independent phase is to construct a data structure which works with all possible edge weights. Therefore, if edge weights change, only the (typically much faster) customization phase has to be redone. An important realization of this customization paradigm is the Customizable Contraction Hierarchies (CCH) technique. A theoretical analysis of CCH allowed to come up with upper bounds for the  search space size, that is, the number of vertices that have to be considered in a query \cite{Bauer2016}. Those results also  inspired an efficient implementation for practical applications \cite{Dibbelt2016}. 

The goal of this paper is to apply the customization paradigm to HL for the first time and to study the theoretical properties of metric-independent labelings. In particular, we are interested in bounds and approximation algorithms for the maximum label size $\Lmax=\max_{v \in V}|L(v)|$ and the average label size $\Lavg= \sfrac{1}{|V|} \sum_{v \in V} |L(v)|$.

\subsection{Related Work}
The Hub Labeling (HL) approach was introduced by Cohen et al.~\cite{Cohen2003} under the name \emph{$2$-hop labeling}. They proposed to label every vertex with a set of vertices (also called \emph{hops}), such that for any $s,t \in V$, the shortest $s$-$t$-path contains a common hop of $s$ and $t$.
As minimizing the total label size is NP-hard~\cite{Babenko2015}, they gave an algorithm which computes labels that are in total at most a factor of $\mathcal{O}(\log n)$ larger than the minimum labeling, where $n$ denotes the number of vertices. However, their algorithm does not scale to large networks, although the runtime was reduced from $\mathcal{O}(n^5)$ to $\mathcal{O}(n^3 \log n)$ by Delling et al.~\cite{Delling2014a}.
A special case of HL is \emph{Hierarchical Hub Labeling (HHL)}~\cite{Abraham2012}, where we have a global order on the vertices such that the label $L(v)$ of a vertex $v$ only contains vertices of higher order.
Babenko et al.~\cite{Babenko2015} showed that minimizing the total label size (which is equivalent to minimizing $\Lavg$) for HHL is NP-hard and gave two algorithms that achieve an approximation ratio of $\mathcal{O}(\sqrt n \log n)$. They also proved that the total label size of HHL can be larger by a factor of $\Theta(\sqrt n)$ than for general HL.
In practice however, heuristic implementations of HHL scale much better than (H)HL approximation algorithms and produce labels of comparable size~\cite{Delling2014,Delling2014a}.
In fact, many of the most successful shortest path algorithms such as Pruned Highway Labeling  \cite{Akiba2013}, Hierarchical 2-Hop labeling   \cite{Ouyang2018}, and Projected vertex separator based 2-Hop labeling \cite{Chen2021} are based on HHL.

One very successful heuristic for HHL computation is based on Contraction Hierarchies (CH),  another  route planning technique that also relies on preprocessing~\cite{Geisberger2012}. There, the idea is to construct an overlay graph, based on hierarchical vertex contraction. In the resulting graph, queries can be answered by a bidirectional run of Dijkstra's algorithm, which only relaxes edges towards vertices higher up in the hierarchy. The set of vertices visited by the Dijkstra run from a vertex $v$ is called the search space $\searchspace(v)$. It was shown in \cite{Abraham2012} that a valid HHL can be constructed by setting $L(v)=\searchspace(v)$ for all vertices in the graph. In practical implementations, the labels derived from CH  are  pruned in a post-processing phase, as they may contain a significant portion of vertices not necessary to fulfill the cover property \cite{Abraham2011a,Abraham2012}. 

Cohen et al.\ \cite{Cohen2003} showed how to compute lower bounds on the average label size $\Lavg$ of HL, based on the \emph{efficiency} of vertex pairs. Their result implies that there are graphs where we have $\Lavg \in \Omega(\sqrt m)$, where $m$ is the number of edges in the input graph.
Rupp and Funke \cite{Rupp2021} showed a lower bound of $\Omega(\sqrt n)$ on the average label size $\Lavg$ of HL for a special family of grid graphs, which are planar and have bounded degree. Moreover, they presented an algorithm which constructs instance-based lower bounds for $\Lavg$ and proved that their approach generates tight lower bounds on different graph classes, including ternary trees.

Recently, approaches for dynamic HL were proposed that aim at updating the labels in case of edge cost changes \cite{Akiba2014,d2019fully,farhan2021fast,farhan2021efficient}. However, these are designed to deal with the increase or decrease of a single edge weight, while in practice there are often a multitude of non-local edge weight changes to be considered at once. For such scenarios, the concept of customizability -- where all edge costs may be changed simultaneously -- is more appropriate.

\subsection{Contribution}
In this paper, we introduce and study the concept of Customizable Hub Labeling (CuHL). We also consider the important special case of hierarchical Customizable Hub Labeling (HCuHL) inspired by the practical usefulness of HHL. We reveal several interesting properties of  CuHL and HCuHL, and provide efficient algorithms for their computation and customization:

\begin{itemize}%
\item We show  that  HCuHL are closely related to CCH. Indeed, we certify that the minimum average and maximum search space sizes in any CCH are equal to the minimum average and maximum label sizes in any HCuHL. We prove that the same does not apply to conventional CH and HHL by providing  an example graph in which the average CH search space size is larger than the average HHL label by a factor of $\Theta(\sqrt n)$. 
\item We argue that the average label size of CuHL is lower bounded by $\sfrac{1}{4}\cdot b_{\sfrac{2}{3}}$, where $b_{\sfrac{2}{3}}$ is the size of a smallest $\sfrac{2}{3}$-balanced separator of the input graph. For HCuHL, we show a slightly larger lower bound of  $\sfrac{2}{3} \cdot b_{\sfrac{2}{3}}$.
\item Based on these obtained lower bounds and further structural insights, we show that we can exploit an approximation framework designed for CCH \cite{Blum2020} to derive a $\mathcal{O}(\log^2 n)$-approximation algorithm for the average label sizes of both CuHL and HCuHL. It is notable that HCuHL allows a polylogarithmic approximation factor, as the best approximation factor for HHL known so far is $\mathcal{O}(\sqrt n \log n)$. 
\item In previous work, it was shown that there can be a gap of $\Omega(\sqrt n)$ between the optimal average label sizes of HL and HHL~\cite{Babenko2015}. We prove that for their customizable counterparts, the gap is in $\mathcal{O}(\log n)$ and hence significantly smaller.
\item We design and analyze efficient customization algorithms  for HCuHL and CuHL. While we heavily rely on the CCH customization framework for HCuHL, we develop a novel and independent algorithm for CuHL.
\end{itemize}

We focus on a theoretical analysis of (H)CuHL. Still, we expect (H)CuHL to outperform other approaches such as CCH in practice, as the query time of HCuHL is linear in the label size, while for CCH, it can be be quadratric in the search space size.

\section{Customizable Hub Labels}
We first formally introduce the concept of metric-independent labelings, analyze their basic properties, and then study their relationship to Customizable Contraction Hierarchies (CCH).

\subsection{Definitions and Properties}
Formally, a \emph{labeling} of a graph $G=(V,E)$ is a function $L \colon V \rightarrow 2^V$, which assigns a label $L(v) \subseteq V$ to every vertex. A Hub Labeling (HL) of a graph with positive edge weights is a labeling which fulfills the \emph{cover property}, i.e.\ for any two vertices $s,t \in V$, the intersection of the labels $L(s) \cap L(t)$ contains some vertex on a shortest $s$-$t$-path.
A labeling $L$ is a \emph{hierarchical} Hub Labeling (HHL), if there is some vertex order $\pi \colon V \rightarrow \{1,\dots, n\}$ such that $u \in L(v)$ implies $\pi(u) \geq \pi(v)$. In this case we say that $L$ respects $\pi$. For a vertex $v$, we call $\pi(v)$ also the \emph{rank} of $v$.
Given some vertex order $\pi \colon V \rightarrow \{1, \dots, n\}$, the \emph{canonical} HHL respecting $\pi$ is the labeling $L^*$ that satisfies $u \in L^*(v)$ if and only if $u$ has maximum rank among all vertices on any shortest $v$-$u$-path. It is known that for a fixed order $\pi$, the canonical HHL $L^*$ is the minimum HHL that respects $\pi$, i.e.\ for any HHL $L$ respecting $\pi$ we have $L(v) \supseteq L^*(v)$~\cite{Babenko2015}. This implies that for any labeling $L$ of a weighted graph $G=(V,E)$, for all $s,t \in V$ the intersection $L(s) \cap L(t)$ contains the vertex $w$ of maximum rank on all shortest $s$-$t$-paths.

We now want to extend those concepts such that they do not depend on one particular metric but work with any metric; i.e. the cover property needs to be fulfilled no matter how the edge costs are chosen later on.
\begin{definition}[Customizable Hub Labeling (CuHL)]
A CuHL of a graph $G=(V,E)$ is a labeling $L$ such that the \emph{customizable cover property} is fulfilled; i.e., for any two vertices $s,t \in V$  and any $s$-$t$-path $P$, the set $L(s) \cap L(t)$ contains some vertex on $P$.
\end{definition}
A \emph{hierarchical} Customizable Hub Labeling (HCuHL) is a CuHL respecting some vertex order $\pi$. The \emph{canonical} HCuHL respecting $\pi$ is the labeling $L^*$ such that $u \in L^*(v)$ if and only if there is some $v$-$u$-path on which $u$ has maximum rank.
In \cite{Babenko2015}, it was proven that  for a given vertex order $\pi$, the canonical HHL is the minimum HHL that respects $\pi$. We can show that this result transfers to the customizable setting.
\begin{restatable}{lemma}{canonicalHCuHLminimum}
For any order $\pi$ the canonical HCuHL is the minimum HCuHL that respects $\pi$.
\end{restatable}
\begin{proof}
To show that the canonical HCuHL $L^*$ fulfills the customizable cover property, consider a simple path $P$ between vertices two $s$ and $t$ and let $w$ be the vertex of maximum rank on $P$. As $w$ is also the vertex of maximum rank on the subpath from $s$ to $w$, it follows that $w \in L^*(s)$, and analogously we have $w \in L^*(t)$, which means that $w \in L^*(s) \cap L^*(t)$.

Consider now an arbitrary HCuHL $L$ that respects $\pi$. We show that $L^*(v) \subseteq L(v)$ for every vertex $v$. Let $w \in L^*(v)$. This means that there is a simple $v$-$w$-path $P$, on which $w$ has maximum rank. As $L(v) \cap L(w)$ needs to contain some vertex on $P$ and $L(w)$ contains only vertices of rank at least $\pi(w)$, it follows that $w \in L(v) \cap L(w)$, and hence $w \in L(v)$.
\end{proof}
If the customizable cover property is fulfilled, then for any metric $\ell \colon E \rightarrow \mathbb{R}$, the traditional cover property is fulfilled on the associated weighted graph.
Still, to be able to answer shortest path queries, the metric $\ell$ needs to be incorporated, i.e., in a customization step the respective shortest path distances need to be assigned to each vertex in a label. In \cref{sec:customization}, we discuss different approaches for label customization in more detail.
Given a customized (H)CuHL, the standard HL query algorithm can be used to compute shortest paths. Hence, after the customization, only \emph{one} common hub is required to correctly answer an $s$-$t$-query.
This means that in the customized CuHL, smaller labels may suffice to ensure correct queries.
However, the following lemma shows that it is not possible to prune the labels in a canonical HCuHL before the customization; i.e., there always exists edge weights such that all labels of the HCuHL are  also needed in the respective HHL  respecting the same vertex order.

\begin{lemma}\label{lem:HCuHL_equal_HHL}
Let $G=(V,E)$ be an unweighted graph. For any vertex order $\pi \colon V \rightarrow \{1,\dots,n\}$ there exist metric edge weights $\ell \colon E \rightarrow \mathbb{N}$ such that the canonical HCuHL is equal to the canonical HHL of the weighted graph $(G,\ell)$.
\end{lemma}

\begin{proof}
Consider some graph $G=(V,E)$ and fix some vertex order $\pi$. Choose the weight of an edge $\{u,v\}$ as $\ell(\{u,v\}) = 3^{\max\{\pi(u),\pi(v)\}}$.
Let $L$ be the canonical HCuHL respecting $\pi$ and let $L'$ be the canonical HHL of the weighted graph $(G,\ell)$ respecting $\pi$.
Clearly, for any vertex $v$ we have $L'(v) \subseteq L(v)$.
Hence, it suffices to show that $L(v) \subseteq L'(v)$ to prove the lemma. Now, consider some vertex $v$ and let $u \in L(v)$. This means that there exists some $v$-$u$-path $P$ such that $u$ has maximum rank on $P$. Let $r = \pi(u)$ be the rank of $u$.
As $P$ is a simple path, no edge weight along $P$ occurs more than twice and $P$ contains exactly one edge of weight $3^r$. This means that the length of $P$ is at most \[ \ell(P) \leq 3^r + 2 \sum_{i=0}^{r-1} 3^i = 3^r + 2 \cdot \sfrac{1}{2} (3^r - 1) = 3^{r+1} - 1.\]

Suppose that we have $u \not \in L'(v)$. This means that there is a shortest $v$-$u$-path $P'$ which contains a vertex $w$ of rank $\pi(w)>r$. This implies that the path $P'$ has length at least $\ell(P') \geq 3^{r+1} > \ell(P)$, which is not possible if $P'$ is a shortest path.
It follows that $u \in L(v)$ and hence $L(v) \subseteq L'(v)$.
\end{proof}

\subsection{Relationship to Customizable Contraction Hierarchies}
Contraction Hierarchies (CH) \cite{Geisberger2012} are one of the most widely used  approaches for shortest path computation in road networks. Given a graph $G$ with positive edge weights and a vertex order $\pi$, the CH graph $G^+$ is obtained from $G$ by adding shortcut edges as follows. There is a shortcut edge between vertices $v$ and $w$ if and only if $G$ contains a shortest $v$-$w$-path $P$ such that for any vertex $u \in P \setminus \{v,w\}$ we have $\pi(u) < \min\{\pi(v),\pi(w)\}$. The weight of the shortcut edge $\{v,w\}$ is chosen as the weight of $P$.
To compute a shortest path between two vertices $s$ and $t$ of $G$, it suffices to perform a bidirectional run of Dijkstra's algorithm from $s$ and $t$ with the restriction that edges are only relaxed if they point ``upwards'' w.r.t.\ the vertex order $\pi$.

In a Customizable Contraction Hierarchy (CCH) \cite{Dibbelt2016}, a shortcut edge $\{v,w\}$ is added if and only if there exists a simple $v$-$w$-path, which (except for $v$ and $w$) only contains vertices $u$ satisfying $\pi(u) < \min\{\pi(v),\pi(w)\}$. Moreover, in a customization step, we can propagate edge weights of the original graph $G$ to the shortcuts. If the weight of an edge of the original graph $G$ changes, only the customization phase has to be repeated (which can be done quite efficiently), but not the construction of the CCH graph $G^*$.

We now show that CCH are indeed closely related to CuHL. In particular, we prove that for any CCH graph $G^*$ there is a hierarchical CuHL $L^*$ such that the search spaces of $G$ coincide with the labels of $L^*$ and vice versa. Abraham et al.\ \cite{Abraham2011a} observed that we can construct a valid HL by choosing CH search spaces as the individual labels.
This approach is also valid in the customizable setting:
If $G^*$ is a CCH graph of some graph $G$, we may choose the label $L^*(v) = \searchspace(v)$ for every vertex $v$, where $\searchspace(v)$ is the search space of $v$ in $G^*$, i.e.\ the set of all vertices which can be reached from $v$ on some path which is increasing w.r.t.\ the vertex order used $\pi$ during the construction of $G^*$.
In fact, this approach yields the canonical HCuHL respecting $\pi$.

\begin{lemma}\label{lem:HCuHL_equal_CCH}
Let $\pi \colon V \rightarrow \{1, \dots, n\}$ be a vertex order of a graph $G=(V,E)$ and let $G^*$ be the resulting CCH graph.
If $L^*$ is the canonical HCuHL respecting $\pi$, then for every vertex $v \in V$ we have $L^*(v) = \searchspace(v)$, where $\searchspace(v)$ denotes the search space of $v$ in $G^*$.
\end{lemma}
\begin{proof}
Let $u \in L^*(v)$, which means that $G$ contains some $v$-$u$-path $P$ such that $u$ has maximum rank on $P$.
We have to show that $G^*$ contains some $v$-$u$-path $P'$ that is increasing w.r.t.\ $\pi$.
Assume that $P$ is not increasing w.r.t.\ $\pi$. Let $x$ and $y$ be the last two consecutive vertices on $P$ satisfying $\pi(y) < \pi(x)$. Let $z$ be the vertex that follows $y$ on $P$. It holds that $\pi(y) < \pi(z)$ and hence, the CCH graph contains the edge $\{x,z\}$. This means that in $P$ we can replace the subpath $x,y,z$ by $x,z$. Note that $x,z$ might also be decreasing w.r.t.\ $\pi$. However, in this case we can just repeat this ``shortcutting'' and as the last vertex on $P$ has maximum rank, we eventually obtain $x,z$ such that $\pi(x) < \pi(z)$. By applying the argument iteratively for all edges $\{x,y\}$ of $P$ that are decreasing, we obtain some path $P'$ in $G^*$ that is increasing w.r.t.\ $\pi$.

Suppose now that $u \in \searchspace(v)$. This means that $G^*$ contains some $v$-$u$-path that is increasing w.r.t.\ $\pi$. By construction of $G^*$, it follows that $G$ contains some $v$-$u$-path such that $u$ has maximum rank. Hence we have $u \in L^*(v)$, which shows that $L^*(v) = \searchspace(v)$.
\end{proof}

It follows that for any CCH-graph $G^*$, there is a HCuHL $L^*$, such that the search space sizes of $G^*$ and the label sizes of $L^*$ coincide.
Moreover, for every HCuHL $L^*$ there is some vertex order $\pi$ respected by $L^*$, which can be used to construct a CCH-graph $G^*$, such that for any vertex $v$ we have $\searchspace(v) \subseteq L^*(v)$.
This means that the minimum average and maximum search space sizes in any CCH are equal to the minimum average and maximum label sizes in any HCuHL, respectively.

\begin{corollary}\label{cor:hcuhl-cch}
In any graph $G$, the optimal average label size of HCuHL is equal to the optimal average search space size of CCH, and the optimal maximum label size of HCuHL is equal to the optimal maximum search space size of CCH.
\end{corollary}

Moreover, it follows from \cref{lem:HCuHL_equal_HHL,lem:HCuHL_equal_CCH} that there are graphs, on which traditional HHL and CCH coincide.
\begin{corollary}
For any graph $G=(V,E)$ and any vertex order $\pi \colon V \rightarrow \mathbb{N}$ there are edge weights $\ell$ such that for the canonical HHL $L$ of $(G,\ell)$ respecting $\pi$ there is a CCH such that for all $v \in V$ we have $L(v) = \searchspace(v)$.
\end{corollary}

Interestingly, the analogue of \cref{lem:HCuHL_equal_CCH} does not hold for traditional HHL and CH. 
As we already mentioned, traditional Contraction Hierarchies can be used to construct a hierarchical Hub Labeling of a given weighted graph by choosing the label of every vertex $v$ as $L(v) = \searchspace(v)$, i.e., the set of vertices that can be reached in the CH graph $G^+$ on a path which is increasing w.r.t.\ the contraction order $\pi$~\cite{Abraham2011a}.
In practice however, the resulting labels can still be pruned:
To guarantee the cover property, it suffices that $L(v)$ contains all vertices $w$ such that $w$ has maximum rank on any \emph{shortest} $v$-$w$-path.
The set of these vertices $w$ is also known as the \emph{direct search space} $\DSS(v)$, and we can observe that it is identical to the label $L(v)$ in the canonical HHL $L$ respecting the contraction order $\pi$.

However, we will now show that there can be a large gap of $\Omega(\sqrt n)$ between the optimal label sizes of an HHL and the optimal CH search space sizes. This is a substantial contrast to the customizable setting. Moreover, it is the very first separation bound for traditional CH and HHL.

\begin{theorem}
There is a graph family for which the minimum average label size of HHL is $\Omega(\sqrt n)$ smaller than the minimum average search space size of CH.
\end{theorem}

\begin{proof}
Consider the following graph (cf.\ \cref{fig:hhl_ch}). It contains $k$ stars with $k$ leaves, the length of every star edge is $1$. The centers of the stars form a clique of size $k$, the clique edges have length $5$. Moreover, there is one additional vertex $s$, which is connected to every star leaf through an edge of length $2$ and to every star center through an edge of length $3$.

\begin{figure}[t]
\centering
\includegraphics{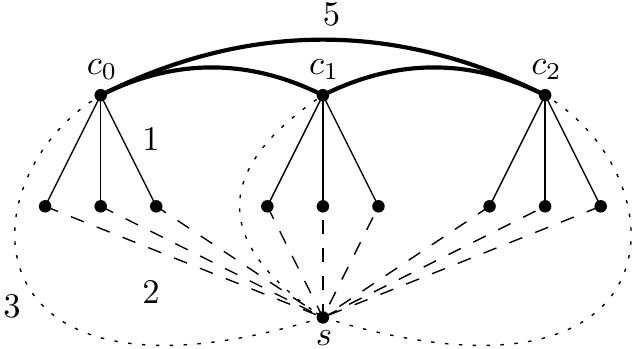}
\caption{Example of a graph family which allows hierarchical labels of average size $\mathcal{O}(1)$, but for which the average search space size of any CH is in $\Omega(\sqrt n)$.}
\label{fig:hhl_ch}
\end{figure}

Consider the following HHL $L$. Denote the star centers by $c_0, \dots, c_{k-1}$. We choose $L(s) = \{s\}$, $L(c_i) = \{s\} \cup \{ c_j \mid j \geq i \}$, and if $v$ is a leaf of a star with center $c_i$, we choose $L(v) = \{s, c_i\}$. It is easy to verify that $L$ satisfies the cover property and that is is hierarchical, as it respects any order $\pi$ satisfying $\pi(s) = n$ and $\pi(c_i) = n-k+i$ .
The label size the $i$-th star center is $|L(c_i)| = k-i+2 \leq k+1$, the label of $s$ has size $2$ and the label of every of the $k^2$ leaves has size $2$.
This means that the average label size of $L$ is $\Lavg = (1 + k^2 \cdot 2 + k \cdot (k+1)) / (1 + k + k^2) \in \mathcal{O}(1)$.

Consider now some CH of the given graph. We show that the average search space size is $\Savg \in \Omega(k) = \Omega(\sqrt n)$.
Let $c_0, \dots, c_{k-1}$ be the star centers w.r.t.\ the contraction order $\pi$, i.e. $\pi(c_0) < \pi(c_1) < \dots < \pi(c_{k-1})$.
Consider some star center $c_i$ where $i \leq \sfrac{k}{2}$. As the star centers form a clique, it follows that the search space of $c_i$ contains all $c_j$ such that $j \geq i$, i.e. $\{c_i, \dots, c_{k-1}\} \subseteq \searchspace(c_i)$. This means that $|\searchspace(c_i)| \geq \sfrac{k}{2}$.

Fix some $i \leq \sfrac{k}{2}-1$ and let $X_i$ be the leaves of the star with center $c_i$.
Let $X_i^<$ and $X_i^>$ be the leaves that are before and after $c_i$ in the contraction order, respectively.
For every $v \in X_i^<$ we have $\searchspace(v) \subseteq \searchspace(c_i)$, and hence $|SS(v)| \geq \sfrac{k}{2}$.
Consider now two leaves $u, v \in X_i^>$. As the shortest $u$-$v$-path is $(u,c_i,v)$ and $c_i$ was contracted before both $u$ and $v$, there must be a shortcut between $u$ and $v$.
This means that either $u \in SS(v)$ or $v \in SS(u)$.
It follows that $\sum_{u \in X_i^>} |SS(u)| \geq {{|X_i^>|}\choose{2}}$.
We obtain that the total search space size of the leaves $X_i$ is
\[
\sum_{u \in X_i} |SS(u)| = \sum_{u \in X_i^<} |SS(u)| + \sum_{u \in X_i^>} |SS(u)| \geq |X_i^<| \cdot \sfrac{k}{2} + {{|X_i^>|}\choose{2}}
\]
As we have $|X_i^<| + |X_i^>| = k$, it follows that $\sum_{u \in X_i} |SS(u)| \in \Omega(k^2)$.
Moreover, there are $\lfloor \sfrac{k}{2} \rfloor$ stars such that $i \leq \sfrac{k}{2}-1$, which implies that $\sum_{v \in V} |SS(u)| \in \Omega(k^3)$.
As the total number of vertices is $(1 + k + k^2) \in \Theta(k^2)$, it follows that the average search space is $\Savg \in \Omega(k^3 / k^2) = \Omega(k) = \Omega(\sqrt n)$.
This concludes the proof.
\end{proof}
Note that the shortest path  between any two vertices in the graph given in Figure \ref{fig:hhl_ch} is unique. Shortest path uniqueness is often assumed to simplify the analysis of algorithms, or to make certain label selection methods viable in the first place (see e.g. \cite{Abraham2012,Delling2014a}). If we allow the existence of multiple shortest paths between two vertices in the input graph, we can increase the gap between CH and HHL even further. In particular, given a complete graph on $n$ vertices $v_1, \dots, v_n$, in which all edge costs are $2$, and one additional node $w$ which is connected to all other vertices through an edge of cost $1$, the average search space size of any CH is in $\Theta(n)$. At the same time, choosing $L(v_i)=\{v_i,w\}$ and $L(w)=\{w\}$ yields an HHL with an average label size of $\mathcal{O}(1)$.

\section{Balanced Separators and Average Label Size}
In this section we investigate the connection between balanced separators and the label sizes of CuHL.
For a graph $G=(V,E)$ and $\alpha \in [0,1]$, an $\alpha$-balanced separator is a set of vertices $S \subseteq V$ such that every connected component of $G[V \setminus S]$ has size at most $\alpha \cdot |V|$. We denote the minimum size of an $\alpha$-balanced separator by $b_{\alpha}$. 

It was previously shown that the average search space size $\Savg$ of any CCH is at least $\Omega(b_{\alpha})$ for $\alpha \geq \sfrac{2}{3}$, and based on this, it was shown that the minimum average search space size can be approximated by a factor of $\mathcal{O}(\log^2 n)$~\cite{Blum2020}.
It follows from \cref{cor:hcuhl-cch} that these results transfer to HCuHL.
For completeness we provide an alternative (slightly simpler) proof for the fact that the label sizes in a HCuHL are in $\Omega(b_\alpha)$. 

Our main focus, though, is to come up with an approximation algorithm for general CuHL. Indeed, we prove that for CuHL we can also bound the average label size in terms of $b_{\alpha}$, although the bound increases by a constant factor compared to the hierarchical case. With the help of this lower bound, we show that we can approximate the average label size of general CuHL within a factor of $\mathcal{O}(\log^2 n)$ as well. 

\subsection{Lower Bounds}
Similar to \cite{Blum2020}, we can show that the average label size of HCuHL is $\Lavg \geq \alpha \cdot b_\alpha$ for $\alpha \in [\sfrac{2}{3},1)$. Moreover, for general CuHL we prove that $\Lavg \geq \frac{2\alpha -1}{2\alpha} \cdot b_\alpha$ if $\alpha \in [\sfrac{1}{2},1)$.
\begin{restatable}{lemma}{lowerboundHCuHL}\label{lem:lower-bound-HCuHL}
For any HCuHL $L$ and any $\alpha \geq \sfrac{2}{3}$, there are $\alpha n$ vertices $v$ such that $|L(v)| \geq b_{\alpha}$.
\end{restatable}
\begin{proof}
Consider some HCuHL $L$ of a graph $G=(V,E)$ and let $\pi$ be a vertex order respected by $L$.
Let $k$ be the smallest integer such that $G \setminus \{ \pi^{-1}(n), \dots, \pi^{-1}(k) \}$ contains a connected component $C$ of size $|C| > \alpha n$.
If we also remove the vertex $\pi^{-1}(k-1)$, $C$ is split into connected components $C_1, \dots, C_r$, which all have size at most $\alpha n$. As it holds that $\alpha n < |C| = 1 + \sum_{i=1}^r |C_i|$, we can choose a subset of these $r$ components such that their union yields a subgraph $C^*$ of size $\sfrac{1}{2} \cdot \alpha n < |C^*| \leq \alpha n$. This means that there are at most $(1 - \sfrac{\alpha}{2}) n \leq \alpha n$ many vertices that are not contained in $C^*$, as we have $\alpha \geq \sfrac{2}{3}$.

Consider now the vertices $S \subseteq \{\pi^{-1}(1), \dots, \pi^{-1}(k-1)\}$ that are adjacent to $C^*$.
Our previous observations imply that $S$ is an $\alpha$-balanced separator, and hence we have $|S| \geq b_\alpha$.
It remains to show that for every $v \in C$ we have $L(v) \subseteq S$.
Let $v \in C$ and $w \in S$. As $C$ is connected and $w$ is adjacent to $C$, there exists some $v$-$w$-path $P$ which is fully contained in the set $C \cup \{w\}$. As $\pi(w) \geq k-1$, and the vertices from $C$ have rank at most $k-1$, it follows that $w$ has maximum rank along $P$, so we have $w \in L(v)$, which concludes the proof.
\end{proof}

It follows immediately for $\alpha \geq \sfrac{2}{3}$ that the average label size of any HCuHL is $\Lavg \geq \alpha \cdot b_{\alpha}$ and moreover, one can show that $\Lavg \geq \sfrac{1}{12} \cdot b_{\sfrac{1}{2}}$.\footnote{Note that in \cite{Blum2020}, it was stated that the average search space size of any CCH is $\Savg \geq \sfrac{2}{9} \cdot b_{\sfrac{1}{2}}$. The presented proof claimed that after removing a minimum $\sfrac{2}{3}$-balanced separator, a connected component of size $n' \geq n / 3$ remains. We are however only guaranteed that $n' \geq \lfloor n / 3 \rfloor$, which yields $\Savg \geq \sfrac{1}{12} \cdot b_{\sfrac{1}{2}}$.}

Let us now consider general CuHL. We first show that we can also bound the average label size of any CuHL in terms of the size of a minimum balanced separator. However, as the proof for HCuHL relies heavily on the hierarchical structure, we need a different argument.

\begin{lemma}\label{lem:lower-bound-CuHL}
For any CuHL $L$ and $\alpha \geq \sfrac{1}{2}$ there are $\frac{(2 \alpha - 1)}{(2 \alpha)} \cdot n$ vertices $v$ such that $\card{L(v)} \geq b_{\alpha}$.
\end{lemma}

\begin{proof}
Consider two vertices $a$ and $b$ of the given graph $G=(V,E)$. It holds that the two vertices are separated by $L(a) \cap L(b)$, i.e. either $G[V \setminus (L(a) \cap L(b))]$  has two distinct connected components $C_a$ and $C_b$ which contain $a$ and $b$, respectively, or $a$ or $b$ (or both) are contained in $L(a) \cap L(b)$. In the latter case, we choose $C_v = \emptyset$, if $v \in L(a) \cap L(b)$ for $v = a$ or $v=b$.
Based on the sizes of $C_a$ and $C_b$, we create a relation $\prec$: If $|C_a| < |C_b|$, let $a \prec b$, if $|C_b| < |C_a|$, let $b \prec a$, and if $|C_a| = |C_b|$ we arbitrarily choose $a \prec b$ or $b \prec a$. 
If $a \prec b$, we also say that $b$ \emph{dwarfs} $a$. 
Note that $a \prec b$ implies $|C_a| \leq \sfrac{n}{2}$.

We claim that the number of vertices $b$ which dwarf at least $(1-\alpha) n$ vertices each is $n' \geq \frac{(2 \alpha - 1)}{(2 \alpha)} \cdot n$.
To prove this, observe that there are ${n \choose 2} \leq \sfrac{1}{2} \cdot n^2$ pairs $(a,b)$ such that $a \prec b$.
If a vertex $b$ dwarfs less than $(1-\alpha) n$ vertices, it participates in less than $(1-\alpha) n$ such pairs $(a,b)$, otherwise it participates in at most $n$ such pairs.
This means that $n' \cdot n + (n-n')  (1-\alpha) n \geq \sfrac{1}{2} \cdot n^2$, which can be rearranged to $n' \geq \frac{(2 \alpha - 1)}{(2 \alpha)} \cdot n$. This proves our claim.

Consider now a vertex $b$ that dwarfs $k \geq (1-\alpha) n$ vertices $a_1, \dots, a_k$.
We show that the label of $b$ has size at least $\card{L(b)} \geq b_{\alpha}$.
This proves the lemma.
For $i = 1, \dots, k$ let $S_i = L(a_i) \cap L(b)$ and if $a_i \not \in S_i$, then denote the connected component of $G[V \setminus S_i]$ containing $a_i$ by $C_i$, otherwise let $C_i = \emptyset$.
Suppose first that there is some $i$ such that $\card{C_i} \geq (1-\alpha) n$.
This means that any other connected component $C$ of $G[V \setminus S_i]$ has size $\card{C} \leq \alpha \cdot n$.
Moreover, as $a_i \prec b$, it holds that $\card{C_i} \leq \sfrac{n}{2} \leq \alpha \cdot n$ for $\alpha \geq \sfrac{1}{2}$. This means that $S_i = L(a) \cap L(b)$ is an $\alpha$-balanced separator, and hence we have $|L(b)| \geq |S_i| \geq b_{\alpha}$.

Suppose now that for any $i$, the connected component $C_i$ of $G[V \setminus S_i]$ has size $\card{C_i} < (1-\alpha) n$. We show that $S = S_1 \cup \dots \cup S_k$ is an $\alpha$-balanced separator. To that end consider a connected component $C$ of the graph $G[V \setminus S]$. If $C$ contains the vertex $a_i$ for some $i$, then it holds that $C \subseteq C_i$, as any connected component of $G[V \setminus S]$ is completely contained in a unique connected component of $G[V \setminus S_i]$. This implies $\card{C} \leq \card{C_i} \leq \sfrac{n}{2} \leq \alpha n$ for $\alpha \geq \sfrac{1}{2}$.
Consider the case that $a_i \not \in C$ for all $i = 1, \dots, k$. This means that $C$ has size $\card{C} \leq n - k \leq n - (1-\alpha) n = \alpha n$. It follows that $S$ is an $\alpha$-balanced separator and as $S = S_1 \cup \dots \cup S_k = L(b) \cap (L(a_i) \cup \dots L(a_k)) \subseteq L(b)$ we obtain that $\card{L(b)} \geq \card{S} \geq b_{\alpha}$.
\end{proof}

This means that the average label size of any CuHL is $\Lavg \geq \frac{(2 \alpha - 1)}{(2 \alpha)} \cdot b_{\alpha}$ for $\alpha \geq \sfrac{1}{2}$.
For $\alpha = \sfrac{2}{3}$ we obtain that $\Lavg \geq \sfrac{1}{4} \cdot b_{\sfrac{2}{3}}$, which is not too far from the lower bound of $\Lavg \geq \sfrac{2}{3} \cdot b_{\sfrac{2}{3}}$ for hierarchical CuHL, which follows from \cref{lem:lower-bound-HCuHL}.

\begin{corollary}\label{cor:Lavg_two_thirds}
For every CuHL, the average label size is $\Lavg \geq \sfrac{1}{4} \cdot b_{\sfrac{2}{3}}$.
\end{corollary}

This result is interesting on its own, as it shows that CuHL is not expected to work well on graphs which do not exhibit small balanced separators. However, the obtained lower bounds on HCuHL and CuHL will also be an important ingredient in the design of suitable approximation algorithms, as discussed in more detail below.

\subsection{Approximation Algorithms}
Based on the lower bounds shown above, we can proceed similarly to \cite{Blum2020} to prove that a HCuHL based on a so-called nested dissection order yields an average label size that is at most a factor of $\bigO(\log n)$ larger than the average label size of any general CuHL. 

A nested dissection order of a graph $G$ is created as follows. We determine a minimum balanced separator $S$ of $G$, remove $S$ from $G$ and recursively process the remaining connected components. This yields a hierarchical decomposition of the graph $G$, the so-called $\alpha$-balanced separator decomposition:
Formally, an $\alpha$-balanced separator decomposition of $G=(V,E)$ is a tree $T$ whose vertices are disjoint subsets of $V$ and that is recursively defined as follows. If $n = 1$, then $T$ consists of a single node $V = S$, otherwise the root of $T$ is an $\alpha$-balanced separator $S$ of $G$ which splits $G$ into connected components $G_1, \dots, G_d$, and moreover, the children of $S$ are the roots of $\alpha$-balanced separator decompositions of $G_1, \dots, G_d$.
To obtain a nested dissection order $\pi_{\nd} \colon V \rightarrow \{1, \dots, n\}$ of $G$, we perform a post-order traversal of $T$, i.e.\ the vertices of the top-level separator are chosen as the top-most vertices, and choose an arbitrary order within every separator.

Denote the canonical HCuHL respecting the nested dissection order $\pi_{\nd}$ by $L^{\nd}$, and let $\Lavg^{\nd}$ be the average label size of $L^{\nd}$. Moreover, denote the minimum average label size of any CuHL of $G$ by $\Lavg^*$.
In the following we use the notation $\Lavg^{\nd}(G)$ and $\Lavg^*(G)$ to refer to the average label sizes of a (sub)graph $G$ when using a nested dissection and an optimal vertex order, respectively.
When we consider subgraphs $G_1, \dots, G_d$ of $G$, we denote the number of their vertices by $n_1, \dots, n_d$, respectively. The following lemma was shown in~\cite{Blum2020}.
\begin{lemma}\label{lem:LavgUB}[Lemma 7 in \cite{Blum2020}].
For $\alpha \geq \sfrac{1}{2}$ let $S$ be a minimum $\alpha$-balanced separator of a graph $G$. If the connected components of $G$ that remain after removing $S$ are $G_1, \dots, G_d$, then we have $\Lavg^{\nd}(G) \leq \sfrac{1}{n} \sum_{i=1}^d n_i \cdot \Lavg^{\nd}(G_i) + \card{S}$.
\end{lemma}

Moreover we can show the following lower bound on the minimum average label size.
\begin{lemma}\label{lem:LavgLB}
If $G_1, \dots, G_d$ are disjoint subgraphs of a graph $G$, then we have $\Lavg^*(G) \geq \sfrac{1}{n} \sum_{i=1}^d n_i \cdot \Lavg^*(G_i)$.
\end{lemma}
\begin{proof}
Let $L$ be a CuHL of $G$ which minimizes $\Lavg$. It holds that $\Lavg^*(G) = \sfrac{1}{n} \sum_{v \in V} \card{L(v)}$.
For every $G_i=(V_i,E_i)$, the labeling $L_i \colon V_i \rightarrow 2^{V_i}$ given by $L_i(v) = L(v) \cap V_i$ is a valid CuHL of $G_i$.
This means that $\sum_{v \in V_i} \card{L_i(v)} \geq n_i \cdot \Lavg^*(G_i)$.
We obtain $\Lavg^*(G) = \sfrac{1}{n} \sum_{v \in V} \card{L(v)} \geq \sfrac{1}{n} \sum_{i=1}^d \sum_{v \in V_i} \card{L(v)} \geq \sfrac{1}{n} \sum_{i=1}^d \sum_{v \in V_i} \card{L_i(v)} \geq \sfrac{1}{n} \sum_{i=1}^d n_i \cdot \Lavg^*(G_i)$.
\end{proof}

By combining \cref{lem:LavgUB,lem:LavgLB} with the lower bound of $\Lavg^* \geq \frac{(2 \alpha - 1)}{(2 \alpha)} \cdot b_{\alpha}$, which follows from \cref{lem:lower-bound-CuHL}, we  now prove that the canonical HCuHL $L^{\nd}$ respecting the nested dissection order $\pi_{\nd}$ approximates the minimum average label size by a logarithmic factor.

\begin{restatable}{theorem}{LavgND}\label{thm:Lavg_ND}
For  $\alpha \geq \sfrac{1}{2}$ and any graph $G$, the canonical HCuHL respecting the optimal nested dissection order $\pi_\nd$ has an average label size of $\Lavg^{\nd} \leq \left(1 + \frac{2 \alpha}{2 \alpha - 1} \log_{\sfrac{1}{\alpha}} n\right) \cdot \Lavg^{*}$.
\end{restatable}
\begin{proof}
Let $T$ be an $\alpha$-balanced separator decomposition which induces the nested dissection order $\pi_\nd$.
For every leaf $S$ of $T$ we have $\Lavg^{\nd}(S) = \Lavg^{*}(S)$ as $S$ contains only one vertex.
Consider now some non-leaf node $S$ of $T$. Let $H$ be the subgraph of $G$ induced by $S$ and its descendants in $T$ and denote the connected components that remain after removing $S$ from $H$ by $H_1, \dots, H_d$.
Denote the number of vertices of $H$ and  of $H_1, \dots, H_d$ by $n'$ and $n_1, \dots, n_d$, respectively. Moreover, assume that for the average label sizes of $H_1, \dots, H_d$ we have an approximation factor of $\gamma$, i.e.\  $\Lavg^{\nd}(H_i) \leq \gamma \cdot \Lavg^{*}(H_i)$. \Cref{lem:LavgUB} implies
\[
	\Lavg^{\nd}(H) \leq \sfrac{1}{n'} \sum_{i=1}^d n_i \cdot \Lavg^{\nd}(H_i) + \vert S \vert
	= \gamma \cdot \sfrac{1}{n'} \sum_{i=1}^d n_i \cdot \Lavg^{*}(H_i) + \vert S \vert
\]
Moreover, $S$ is a minimum $\alpha$-balanced separator of $H$, so it follows from \cref{lem:lower-bound-CuHL} that $\Lavg^{*}(H) \geq \frac{2\alpha - 1}{2\alpha} \cdot \card{S}$, which can be rearranged to $\card{S} \leq \frac{2\alpha}{2\alpha - 1} \cdot \Lavg^{*}(H)$. In combination with \cref{lem:LavgLB} we obtain
\[
	\gamma \cdot \sfrac{1}{n'} \sum_{i=1}^d n_i \cdot \Lavg^{*}(H_i) + \card{S}
	\leq \gamma \cdot \Lavg^{*}(H) + \frac{2\alpha}{2\alpha - 1} \cdot \Lavg^{*}(H) 
	\leq \left(\gamma + \frac{2 \alpha}{2 \alpha - 1}\right) \cdot \Lavg^{*}(H)
\]
As for every leaf $S$ we have $\Lavg^{\nd}(S) = \Lavg^{*}(S)$ and the height of the separator decomposition $T$ is at most $\log_{1/\alpha} n$, it follows by induction that $\Lavg^{\nd}(G) \leq \left(1 + \frac{2 \alpha}{2 \alpha - 1} \log_{1/\alpha} n\right) \cdot \Lavg^{*}(G)$.
\end{proof}
We remark that for $\alpha \geq \sfrac{2}{3}$, using \cref{lem:lower-bound-HCuHL} instead of \cref{lem:lower-bound-CuHL} yields a better approximation factor relative to the optimal HCuHL.
\begin{lemma}
For $\alpha \geq \sfrac{2}{3}$ the canonical HCuHL respecting $\pi_{\nd}$ has an average label size which is at most $1 + \sfrac{1}{\alpha} \log_{\sfrac{1}{\alpha}} n$ times larger than the average label size of any HCuHL.
\end{lemma}
Note that \cref{thm:Lavg_ND} does not immediately imply a polynomial time approximation algorithm for minimizing $\Lavg$ as computing $\alpha$-balanced separators of minimum size is NP-hard~\cite{Feige2006}. However, as it was observed in \cite{Blum2020}, we can use a so-called pseudo-approximation algorithm due to Leighton and Rao \cite{Leighton1988} to compute a $\sfrac{3}{4}$-balanced separator of size $\mathcal{O}(\log n) \cdot b_{\sfrac{2}{3}}$ in polynomial time. This pseudo-approximation algorithm can also be used to compute an $\alpha$-balanced separator decomposition, which induces a nested dissection order $\pi_{\mathrm{pnd}}$.
For the canonical HCuHL respecting $\pi_{\mathrm{pnd}}$ we can show that it approximates the minimum average label size $\Lavg^*$ by a factor of $\mathcal{O}(\log^2 n)$, as compared to the optimal nested dissection order we lose a factor of $\mathcal{O}(\log n)$ for every separator.
We obtain the following theorem:

\begin{theorem}\label{thm:Approx}
For any graph $G$ we can compute a HCuHL in polynomial time which has an average label size of $\mathcal{O}(\log^2 n) \cdot \Lavg^*$, where $\Lavg^*$ denotes the minimum average label size of any CuHL.
\end{theorem}

Note that better approximation ratios can be achieved, if we can find smaller balanced separators than the algorithm of Leighton and Rao.
For instance, for the class of $\sqrt n \times \sqrt n$ grid graphs, it follows from \cite{Bauer2016} that we can compute a nested dissection order in polynomial time, which yields  labels of size at most $3 \sqrt n$.
As every $\sfrac{2}{3}$-balanced separator of a $\sqrt{n} \times \sqrt{n}$ grid graph has size at least $\sqrt{\sfrac{2}{3} \cdot n}$ \cite{Lipton1979}, \cref{cor:Lavg_two_thirds} implies a lower bound of $\sfrac{1}{4} \cdot \sqrt{\sfrac{2}{3} \cdot n}$ on the average label size. 
We hence get a constant approximation factor:\footnote{In \cite{Blum2020}, it was stated, that nested dissection yields an approximation factor of $4.5$ for the average search space size of CCH on grid graphs. To prove this, it was however implicitly assumed that $b_{\sfrac{2}{3}} \geq \sqrt n$, which is too large. Using $b_{\sfrac{2}{3}} \geq \sqrt{\sfrac{2}{3} \cdot n}$ instead yields a lower bound of $3/ (\sfrac{2}{3} \cdot \sqrt{\sfrac{2}{3}}) \approx 5.5$.}
\begin{lemma}
For a $\sqrt{n} \times \sqrt{n}$ grid graph we can compute a HCuHL in polynomial time, whose average label size is at most $3 / (\sfrac{1}{4} \cdot \sqrt{\sfrac{2}{3}}) \approx 14.7$ times larger than the average label size of any CuHL, and at most  $3/(\sfrac{2}{3} \cdot \sqrt{\sfrac{2}{3}}) \approx 5.5$ larger than for any HCuHL.
\end{lemma}

\Cref{thm:Lavg_ND} shows that the gap between the average label size of hierarchical CuHL and general CuHL is at most $\mathcal{O}(\log n)$.
This distinguishes Customizable HL from traditional HL, where hierarchical labelings can be larger than general labelings by a factor of $\Omega(\sqrt n)$~\cite{Babenko2015}.

\begin{corollary}
For any graph $G$, the optimal average label size of hierarchical CuHL is at most $\mathcal{O}(\log n)$ larger than the average label size of any CuHL.
\end{corollary}

\section{Customization Algorithms}\label{sec:customization}
We now describe how the customization step of CuHL can be performed. Given some CuHL $L$ of a graph $G=(V,E)$ and edge weights $\ell \colon E \rightarrow \mathbb{R}_0^+$, we need to compute distance labels $d_v[u]$ for all $u, v$ satisfying $u \in L(v)$, such that we can correctly answer shortest path queries in the resulting weighted graph.
We first explain how we can proceed for hierarchical labels, then we consider general CuHL.

\subsection{Hierarchical CuHL}
For a HCuHL $L$ of a graph $G=(V,E)$, the customization can be performed as follows.
Let $\pi$ be some order respected by $L$ and denote the associated CCH graph by $G^*=(V,E \cup E^+)$.
W.l.o.g.\ we may assume that $L$ is the canonical HCuHL respecting $\pi$.
\Cref{lem:HCuHL_equal_CCH} states that for all $v \in V$ we have $L(v) = \searchspace(v)$.

This means that we can compute the distance labels $d_v[u]$ as in the customization step of CCH.
To that end, let $\Nup(v) = \{u \mid \{v,u\} \in E \cup E^+ \text{ and } \pi(v) > \pi(u) \}$ and $\Ndown(v) = \{u \mid \{v,u\} \in E \cup E^+ \text{ and } \pi(v) < \pi(u) \}$.
For $u \in L(v)$ we store the distance from $v$ to $u$ in $d_v[u]$. Initially we choose $d_v[u] = \ell(v,u)$ if $\{v,u\} \in E$ and $d_v[u] = \infty$ otherwise.
Then we follow the CCH approach of Dibbelt et al.\ \cite{Dibbelt2016} to compute $d_v[u]$ for all $u \in \Nup(v)$, i.e., we first only label edges of the CCH graph.
To that end we iterate over all vertices $v \in V$ increasingly by $\pi(v)$. For each $v$, we iterate over all $u \in \Nup(v)$ and consider all ``lower triangles'' $(v,w,u)$, i.e., all $w \in \Ndown(v) \cap \Ndown(u)$.
For each such $(v,w,u)$ we check whether $d_v[w]+d_u[w] < d_v[u]$ and if this is the case, we decrease $d_v[u]$ to $d_v[w]+d_u[w]$.
The running time of this step is linear in the number of lower triangles.

It remains to compute $d_v[u]$ for all $u \in L(v) \setminus \Nup(v)$, which can be done with the algorithm of Dijkstra. To reduce the runtime, it suffices to consider only ``upwards'' edges, i.e., we visit a neighbor $y$ of a vertex $x$ only if $y \in \Nup(x)$.
This yields a distance label $d_v[u]$ for all $u \in L(v)$, although in general we do not have $d_v[u] = \dist_G(v,u)$. %
Still, the correctness of the query algorithm is not affected, as we are guaranteed that for any query pair $(s,t)$ there is some vertex $p$ on the shortest $s$-$t$-path that can be reached from both $s$ and $t$ on a path in $G^*$ that is increasing w.r.t.\ $\pi$.

The running time of this approach is $\mathcal{O}\left(\sum_{v \in V} \card{L(v)} \log \card{L(v)} + \right. \allowbreak \left. \sum_{w \in L(v)} \card{\Nup(w)}\right)$, as a single Dijkstra run from some vertex $v$ visits $\card{L(v)}$ vertices and $\sum_{w \in L(v)} \card{\Nup(w)}$ edges.
Note also that for an efficient access to $\Nup(v)$ and $\Ndown(v)$, we have to store these sets explicitly in addition to $L(v)$.
As we have $\Nup(v) \subseteq L(v)$ though, we can just add a flag to all $w \in L(v)$ that satisfy $w \in \Nup(v)$.
Moreover, as we have $w \in \Ndown(v)$ if and only if $v \in \Ndown(w)$, additionally storing $\Ndown(v)$ only increases the space consumption by a factor of two.

Alternatively, we can proceed as follows.
Suppose that $u \in L(v) \setminus \Nup(v)$. This means that $v$ has some upper neighbor $w \in \Nup(v)$ such that $\dist(v,u) = \dist(v,w) + \dist(w,u)$. So, provided that $d_w[u]$ has already been computed, we can choose $d_v[u]$ on basis of $d_v[w]$ and $d_w[u]$.
We exploit this fact by iterating over all vertices $u \in V$ decreasingly by $\pi(u)$.
For each $v \in L(v)$, we iterate over all $w \in \Nup(v)$ and choose $d_v[u] \gets \min\{d_v[u], d_v[w]+d_w[u]\}$.
Eventually, $d_v[u]$ has been set for all $v \in V$ and all $u \in L(v)$.
As in the previously described method, this top-down approach does necessarily yield entries that satisfy $d_v[u] = \dist_G(v,u)$, but we are still guaranteed to obtain correct queries.

The running time of the first approach can be bounded by $\bigO\left(\sum_{v \in V} \card{L(v)} \log \card{L(v)} + \right. \allowbreak \left. \sum_{w \in L(v)} \card{\Nup(w)}\right)$, as a single Dijkstra run from some vertex $v$ visits $\card{L(v)}$ vertices and $\sum_{w \in L(v)} \card{\Nup(w)}$ edges.

The running time of the second approach is bounded by $\bigO\left( \sum_{v \in V} \card{L(v)} \cdot \card{\Nup(v)} \right)$.
For vertices $v$ where $\card{\Nup(v)}$ is large, the Dijkstra-based approach might be faster, while we expect the second approach to be faster for vertices further down in the vertex hierarchy.
We remark that it is also possible to combine both approaches, i.e., to use the algorithm of Dijkstra for the vertices of high $\pi$ and the top-down approach for the remaining vertices.
We propose to evaluate different customization strategies practically in future research.

Note also that the described approaches suggest to store the sets $\Nup(v)$ and $\Ndown(v)$ in addition to $L(v)$.
As we have $\Nup(v) \subseteq L(v)$ though, we can just add a flag to all $w \in L(v)$ that satisfy $w \in \Nup(v)$.
Moreover, as we have $w \in \Ndown(v)$ if and only if $v \in \Ndown(w)$, additionally storing $\Ndown(v)$ only increases the space consumption by a factor of two.

\subsection{General CuHL}
Consider a CuHL $L$ of a graph $G=(V,E)$ and edge weights $\ell \colon E \rightarrow \mathbb{R}_0^+$.
Let $u \in L(v)$. For any shortest path that consists of a single edge $\{u,v\}$, we have $\dist(u,v) = \ell(u,v)$.
Therefore, we initialize the distance label $d_v[u]$ as $\ell(u,v)$ if $\{u,v\} \in E$ and as $\infty$ otherwise.
Consider now some shortest $v$-$u$-path $P$ which consists of at least two edges.
This means that we can split $P$ into an edge $v$-$w$ and a non-empty shortest $w$-$u$-path.
It holds that $\dist(v,w) = \ell(v,w) + \dist(w,u)$.
The idea is that whenever the customization affects the answer to a $w$-$u$-query, we check whether an update of $d_v[u]$ is necessary.
By the customizable cover property we know that any $w$-$u$-path contains a vertex $w'$ such that $w' \in L(w) \cap L(u)$.
We distinguish the cases that \inlineItem{(a)} $w' = u$, \inlineItem{(b)} $w' = w$, and \inlineItem{(c)} $w' \not \in \{w,u\}$ (cf.\ \cref{fig:label_structure}).
It holds that the  distance between $w$ and $u$ can only change if $d_w[w']$ or $d_u[w']$ decreases.

\begin{figure}
\begin{subfigure}[t]{.3\textwidth}
\centering
\includegraphics[page=1]{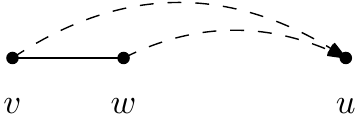}
\caption{$u \in L(w)$}
\end{subfigure}
\hfill
\begin{subfigure}[t]{.3\textwidth}
\centering
\includegraphics[page=2]{customization_label_structure}
\caption{$w \in L(u)$}
\end{subfigure} 
\hfill
\begin{subfigure}[t]{.3\textwidth}
\centering
\includegraphics[page=3]{customization_label_structure}
\caption{$w' \in L(w) \cap L(u)$}
\end{subfigure} 
\caption{The three possibilities when $u \in L(v)$ and the shortest $v$-$u$-path contains at least two edges. Solid lines denote edges in $G$, a dashed edge from $x$ to $y$ indicates that $y \in L(x)$.}\label{fig:label_structure}
\end{figure}

Let us now change the point of view and suppose that for $y \in L(x)$ the distance label $d_x[y]$ is updated. Based on the previous observations we identify all $u,v$ such that $d_v[u]$ needs to be updated.
It might be that case \inlineItem{(a)} applies, i.e., $x$ and $y$ take the roles of $w$ and $u$, respectively.
In this case we have $d_v[y] = \ell(v,x) + d_x[y]$.
This means that we have to iterate over all $v \in N(x)$ and perform the update $d_v[y] \gets \min \{ d_v[y], \ell(v,x) + d_x[y] \}$.
Similar, for case \inlineItem{(b)}, we have to iterate over all $v \in N(y)$ and perform the update $d_v[x] \gets \min \{ d_v[x], \ell(v,y) + d_x[y] \}$.
In case \inlineItem{(c)}, we have $y=w'$ whereas $x$ can take the role of $w$ or $u$.
For the former case we have to iterate over all $v \in N(x)$ and all vertices $u$ in the inverse label $L\inv(y) = \{ z \mid y \in L(z) \}$ of $y$ and perform the update $d_v[u] \gets \{ d_v[u], \ell(v,x) + d_x[y] + d_u[y] \}$.
For the latter case we have to iterate over all $w \in L\inv(y)$ and all $v \in N(w)$ and perform the update $d_v[u] \gets \{ \min d_v[u], d_x[y] + d_w[y] + \ell(v,w) \}$.

To keep track of all the labels that were modified and might trigger further changes, we use a queue $Q$, which initially contains all pairs $(x,y)$ such that $\{x,y\} \in E$ and $y \in L(x)$.
As long as $Q$ is not empty, we retrieve the first element and identify all $(u,v)$ such that $d_u[v]$ needs to be updated. All such pairs $(u,v)$ are added to $Q$, if they are not contained in $Q$ yet.
Whenever the queue $Q$ is empty, the customization is finished.

To analyze the runtime of the customization procedure, we observe that whenever some pair $(x,y)$ is retrieved from the queue $Q$ for the $k$-th time, then $d_x[y]$ is the length of the shortest $x$-$y$-path that contains at most $k$ edges. If $D_{\mathrm{hop}}$ denotes the maximum number of edges of all shortest paths, then it follows that every pair $(x,y)$ is retrieved at most $D_{\mathrm{hop}} +1$ times from $Q$.
The time required to handle a single pair $(x,y)$ can be determined as follows. The updates for cases \inlineItem{(a)} and \inlineItem{(b)} take time $\bigO(\deg(x))$ and $\bigO(\deg(y))$, respectively. The updates for case \inlineItem{(c)} take time $\bigO(\deg(x) \cdot \card{L\inv(y)} + \deg(w) \cdot \card{L\inv(y)})$.
This means that the total runtime is $\bigO(\sum_{v} \card{L(v)} \cdot D_{\mathrm{hop}} \cdot \Delta \cdot \max_{w} \card{L\inv(w)})$.

For this customization strategy, we need to store the sets $\Nup(v), \Ndown(v), L\inv(v)$ in addition to $L(v)$.
However, as we already mentioned in the previous section, we have $\Nup(v) \subseteq L(v)$, and hence, we can store the set $\Nup(v)$ by adding a flag to all $w \in L(v)$ that are contained in $\Nup(v)$.
Moreover, as $w \in \Ndown(v)$ iff $v \in \Nup(w)$ and $w \in L\inv(v)$ iff $v \in L(w)$, additionally storing $\Ndown(v)$ and $L\inv(v)$ only increases the space consumption by a factor of two.

\section{Conclusions and Future Work}
We introduced the concept of Customizable Hub Labeling and studied its theoretical properties. \Cref{fig:overview} illustrates  known results about approximability and relationships of hierarchical and general (Customizable) Hub Labeling and (Customizable) Contraction Hierarchies. There are interesting asymmetries between the traditional and customizable relationships. While we proved that the gap between optimal average CH search spaces and average HHL label sizes can be as large as $\Omega(\sqrt n)$, we also proved that for CCH and HCuHL the respective sizes always coincide.
Still, as the query time of HCuHL is linear in the label size, while for CCH, it can be be quadratric in the search space size, HCuHL is expected to outperform CCH in practice. Furthermore, the gap between HCuHL and CuHL is at most logarithmic, but there are instances with a gap of $\Omega(\sqrt n)$ between HHL and HL. One open question is whether the gap between HCuHL and CuHL can be tightened even further. Also, the known approximation algorithms are far better for HCuHL than for HHL. It is unclear whether the approximation factor for HHL can be improved or whether an inapproximability result could manifest this difference. Finally, it would be interesting to investigate whether the proposed approximation and customization algorithms for (H)CuHL are useful for practical application. While in the non-customizable setting,  HHL are used instead of (potentially far better) general HL mainly  due to practicability and the findings of empirical studies, our  results justify to focus on HCuHL in the customizable setting also from a theoretical perspective.
\begin{figure}[h]
\includegraphics[width=\textwidth]{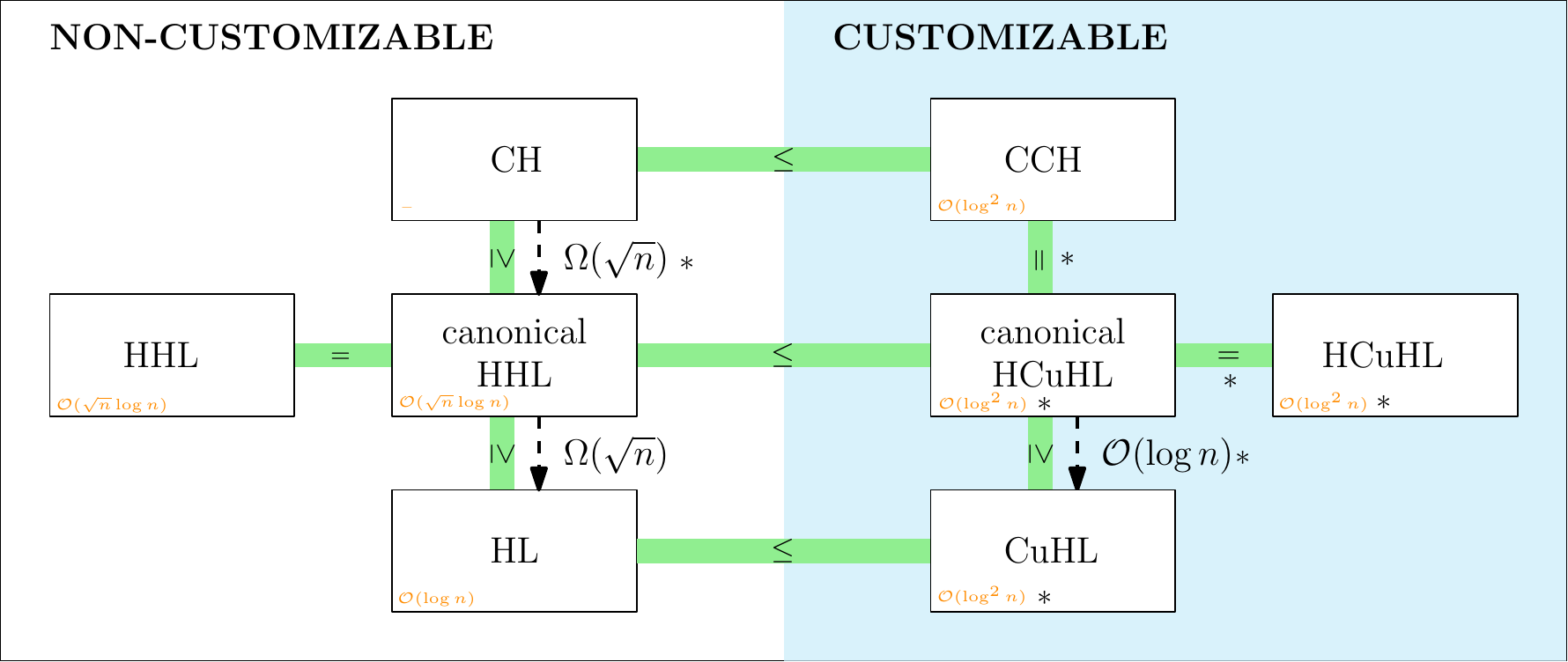}
\caption{Overview of the relationships between optimal average search space sizes and optimal average label sizes for conventional and customizable Contraction Hierarchies and (Hierarchical) Hub Labeling. A dashed arrow indicates that gap between the optimal sizes of the linked techniques is within the given upper/lower bound. Known approximation guarantees are shown via the orange labels in the lower left corner. New results from this paper are marked with a star $*$.}\label{fig:overview} 
\end{figure}
\bibliographystyle{plainurl}%
\bibliography{cuhl.bib}

\begin{thebibliography}{10}

\bibitem{Abraham2011a}
Ittai Abraham, Daniel Delling, Andrew~V. Goldberg, and Renato Fonseca~F.
  Werneck.
\newblock A hub-based labeling algorithm for shortest paths in road networks.
\newblock In Panos~M. Pardalos and Steffen Rebennack, editors, {\em Proc. 10th
  Int. Symp. Experimental Algorithms {(SEA '11)}}, volume 6630 of {\em Lecture
  Notes in Computer Science}, pages 230--241. Springer, 2011.
\newblock \href {https://doi.org/10.1007/978-3-642-20662-7\_20}
  {\path{doi:10.1007/978-3-642-20662-7\_20}}.

\bibitem{Abraham2012}
Ittai Abraham, Daniel Delling, Andrew~V. Goldberg, and Renato Fonseca~F.
  Werneck.
\newblock Hierarchical hub labelings for shortest paths.
\newblock In Leah Epstein and Paolo Ferragina, editors, {\em Proc. 20th Ann.
  Europ. Symp. Algorithms {(ESA '12)}}, volume 7501 of {\em Lecture Notes in
  Computer Science}, pages 24--35. Springer, 2012.
\newblock \href {https://doi.org/10.1007/978-3-642-33090-2\_4}
  {\path{doi:10.1007/978-3-642-33090-2\_4}}.

\bibitem{Akiba2013}
Takuya Akiba, Yoichi Iwata, and Yuichi Yoshida.
\newblock Fast exact shortest-path distance queries on large networks by pruned
  landmark labeling.
\newblock In Kenneth~A. Ross, Divesh Srivastava, and Dimitris Papadias,
  editors, {\em Proc. {ACM} {SIGMOD} Int. Conf. Management of Data {(SIGMOD
  '13)}}, pages 349--360. {ACM}, 2013.
\newblock \href {https://doi.org/10.1145/2463676.2465315}
  {\path{doi:10.1145/2463676.2465315}}.

\bibitem{Akiba2014}
Takuya Akiba, Yoichi Iwata, and Yuichi Yoshida.
\newblock Dynamic and historical shortest-path distance queries on large
  evolving networks by pruned landmark labeling.
\newblock In Chin{-}Wan Chung, Andrei~Z. Broder, Kyuseok Shim, and Torsten
  Suel, editors, {\em Proc. 23rd Int. Conf. World Wide Web {(WWW '14)}}, pages
  237--248. {ACM}, 2014.
\newblock \href {https://doi.org/10.1145/2566486.2568007}
  {\path{doi:10.1145/2566486.2568007}}.

\bibitem{Babenko2015}
Maxim~A. Babenko, Andrew~V. Goldberg, Haim Kaplan, Ruslan Savchenko, and
  Mathias Weller.
\newblock On the complexity of hub labeling (extended abstract).
\newblock In Giuseppe~F. Italiano, Giovanni Pighizzini, and Donald Sannella,
  editors, {\em Proc. 40th Int. Symp. Mathematical Foundations of Computer
  Science {(MFCS '15)}}, volume 9235 of {\em Lecture Notes in Computer
  Science}, pages 62--74. Springer, 2015.
\newblock \href {https://doi.org/10.1007/978-3-662-48054-0\_6}
  {\path{doi:10.1007/978-3-662-48054-0\_6}}.

\bibitem{Bauer2016}
Reinhard Bauer, Tobias Columbus, Ignaz Rutter, and Dorothea Wagner.
\newblock Search-space size in contraction hierarchies.
\newblock {\em Theor. Comput. Sci.}, 645:112--127, 2016.
\newblock \href {https://doi.org/10.1016/j.tcs.2016.07.003}
  {\path{doi:10.1016/j.tcs.2016.07.003}}.

\bibitem{Blum2020}
Johannes Blum and Sabine Storandt.
\newblock Lower bounds and approximation algorithms for search space sizes in
  contraction hierarchies.
\newblock In Fabrizio Grandoni, Grzegorz Herman, and Peter Sanders, editors,
  {\em Proc. 28th Ann. Europ. Symp. Algorithms {(ESA '20)}}, volume 173 of {\em
  LIPIcs}, pages 20:1--20:14. Schloss Dagstuhl - Leibniz-Zentrum f{\"{u}}r
  Informatik, 2020.
\newblock \href {https://doi.org/10.4230/LIPIcs.ESA.2020.20}
  {\path{doi:10.4230/LIPIcs.ESA.2020.20}}.

\bibitem{Chen2021}
Zitong Chen, Ada~Wai{-}Chee Fu, Minhao Jiang, Eric Lo, and Pengfei Zhang.
\newblock {P2H:} efficient distance querying on road networks by projected
  vertex separators.
\newblock In Guoliang Li, Zhanhuai Li, Stratos Idreos, and Divesh Srivastava,
  editors, {\em Proc. 2021 Int. Conf. Management of Data {(SIGMOD '21)}}, pages
  313--325. {ACM}, 2021.
\newblock \href {https://doi.org/10.1145/3448016.3459245}
  {\path{doi:10.1145/3448016.3459245}}.

\bibitem{Cohen2003}
Edith Cohen, Eran Halperin, Haim Kaplan, and Uri Zwick.
\newblock Reachability and distance queries via 2-hop labels.
\newblock {\em {SIAM} J. Comput.}, 32(5):1338--1355, 2003.
\newblock \href {https://doi.org/10.1137/S0097539702403098}
  {\path{doi:10.1137/S0097539702403098}}.

\bibitem{d2019fully}
Gianlorenzo D'angelo, Mattia D'emidio, and Daniele Frigioni.
\newblock Fully dynamic 2-hop cover labeling.
\newblock {\em {ACM} J. Exp. Algorithmics}, 24(1):1--36, 2019.
\newblock \href {https://doi.org/10.1145/3299901} {\path{doi:10.1145/3299901}}.

\bibitem{delling2011customizable}
Daniel Delling, Andrew~V Goldberg, Thomas Pajor, and Renato~F Werneck.
\newblock Customizable route planning.
\newblock In Panos~M. Pardalos and Steffen Rebennack, editors, {\em Proc. 10th
  Int. Symp. Experimental Algorithms {(SEA '11)}}, volume 6630 of {\em Lecture
  Notes in Computer Science}, pages 376--387. Springer, 2011.
\newblock \href {https://doi.org/10.1007/978-3-642-20662-7\_32}
  {\path{doi:10.1007/978-3-642-20662-7\_32}}.

\bibitem{Delling2014}
Daniel Delling, Andrew~V. Goldberg, Thomas Pajor, and Renato~F. Werneck.
\newblock Robust distance queries on massive networks.
\newblock In Andreas~S. Schulz and Dorothea Wagner, editors, {\em Proc. 22th
  Ann. Europ. Symp. Algorithms {(ESA '14)}}, volume 8737 of {\em Lecture Notes
  in Computer Science}, pages 321--333. Springer, 2014.
\newblock \href {https://doi.org/10.1007/978-3-662-44777-2\_27}
  {\path{doi:10.1007/978-3-662-44777-2\_27}}.

\bibitem{Delling2014a}
Daniel Delling, Andrew~V. Goldberg, Ruslan Savchenko, and Renato~F. Werneck.
\newblock Hub labels: Theory and practice.
\newblock In Joachim Gudmundsson and Jyrki Katajainen, editors, {\em Proc. 13th
  Int. Symp. Experimental Algorithms {(SEA '14)}}, volume 8504 of {\em Lecture
  Notes in Computer Science}, pages 259--270. Springer, 2014.
\newblock \href {https://doi.org/10.1007/978-3-319-07959-2_22}
  {\path{doi:10.1007/978-3-319-07959-2_22}}.

\bibitem{delling2013hub}
Daniel Delling, Andrew~V Goldberg, and Renato~F Werneck.
\newblock Hub label compression.
\newblock In Vincenzo Bonifaci, Camil Demetrescu, and Alberto
  Marchetti{-}Spaccamela, editors, {\em Proc. 12th Int. Symp. Experimental
  Algorithms {(SEA '13)}}, volume 7933 of {\em Lecture Notes in Computer
  Science}, pages 18--29. Springer, 2013.
\newblock \href {https://doi.org/10.1007/978-3-642-38527-8\_4}
  {\path{doi:10.1007/978-3-642-38527-8\_4}}.

\bibitem{Dibbelt2016}
Julian Dibbelt, Ben Strasser, and Dorothea Wagner.
\newblock Customizable contraction hierarchies.
\newblock {\em {ACM} J. Exp. Algorithmics}, 21(1):1.5:1--1.5:49, 2016.
\newblock \href {https://doi.org/10.1145/2886843} {\path{doi:10.1145/2886843}}.

\bibitem{farhan2021efficient}
Muhammad Farhan and Qing Wang.
\newblock Efficient maintenance of distance labelling for incremental updates
  in large dynamic graphs.
\newblock In Yannis Velegrakis, Demetris Zeinalipour{-}Yazti, Panos~K.
  Chrysanthis, and Francesco Guerra, editors, {\em Proc. 24th Int. Conf.
  Extending Database Technology {(EDBT '21)}}, pages 385--390.
  OpenProceedings.org, 2021.
\newblock \href {https://doi.org/10.5441/002/edbt.2021.39}
  {\path{doi:10.5441/002/edbt.2021.39}}.

\bibitem{farhan2021fast}
Muhammad Farhan, Qing Wang, Yu~Lin, and Brendan McKay.
\newblock Fast fully dynamic labelling for distance queries.
\newblock {\em The VLDB Journal}, pages 1--24, 2021.
\newblock \href {https://doi.org/10.1007/s00778-021-00707-z}
  {\path{doi:10.1007/s00778-021-00707-z}}.

\bibitem{Feige2006}
Uriel Feige and Mohammad Mahdian.
\newblock Finding small balanced separators.
\newblock In Jon~M. Kleinberg, editor, {\em Proc. 38th Ann. {ACM} Symp. Theory
  of Computing {(STOC '06)}}, pages 375--384. {ACM}, 2006.
\newblock \href {https://doi.org/10.1145/1132516.1132573}
  {\path{doi:10.1145/1132516.1132573}}.

\bibitem{Geisberger2012}
Robert Geisberger, Peter Sanders, Dominik Schultes, and Christian Vetter.
\newblock Exact routing in large road networks using contraction hierarchies.
\newblock {\em Transportation Science}, 46(3):388--404, 2012.
\newblock \href {https://doi.org/10.1287/trsc.1110.0401}
  {\path{doi:10.1287/trsc.1110.0401}}.

\bibitem{Leighton1988}
Frank~Thomson Leighton and Satish Rao.
\newblock An approximate max-flow min-cut theorem for uniform multicommodity
  flow problems with applications to approximation algorithms.
\newblock In {\em Proc. 29th Ann. Symp. Foundations of Computer Science {(FOCS
  '88)}}, pages 422--431. {IEEE} Computer Society, 1988.
\newblock \href {https://doi.org/10.1109/SFCS.1988.21958}
  {\path{doi:10.1109/SFCS.1988.21958}}.

\bibitem{Lipton1979}
Richard~J Lipton and Robert~Endre Tarjan.
\newblock A separator theorem for planar graphs.
\newblock {\em {SIAM} J. Appl. Math.}, 36(2):177--189, 1979.
\newblock \href {https://doi.org/10.1137/0136016} {\path{doi:10.1137/0136016}}.

\bibitem{Ouyang2018}
Dian Ouyang, Lu~Qin, Lijun Chang, Xuemin Lin, Ying Zhang, and Qing Zhu.
\newblock When hierarchy meets 2-hop-labeling: Efficient shortest distance
  queries on road networks.
\newblock In Gautam Das, Christopher~M. Jermaine, and Philip~A. Bernstein,
  editors, {\em Proc. 2018 Int. Conf. Management of Data {(SIGMOD '18)}}, pages
  709--724. {ACM}, 2018.
\newblock \href {https://doi.org/10.1145/3183713.3196913}
  {\path{doi:10.1145/3183713.3196913}}.

\bibitem{Rupp2021}
Tobias Rupp and Stefan Funke.
\newblock A lower bound for the query phase of contraction hierarchies and hub
  labels and a provably optimal instance-based schema.
\newblock {\em Algorithms}, 14(6), 2021.
\newblock \href {https://doi.org/10.3390/a14060164}
  {\path{doi:10.3390/a14060164}}.

\end{thebibliography}

\end{document}